\documentclass[12pt]{article}

\usepackage{mathrsfs}
\usepackage{amsfonts}
\usepackage{tikz}
\usepackage{amssymb,bm}
\usepackage{amsmath}
\usepackage{amsthm}
\usepackage{bbding}
\usepackage{float}

\usetikzlibrary{matrix,shapes,snakes}

\allowdisplaybreaks
\newtheorem{thm}{Theorem}[section]
\newtheorem{defi}[thm]{Definition}
\newtheorem{lem}[thm]{Lemma}

\newtheorem{prop}[thm]{Proposition}
\newtheorem{rmk}[thm]{Remark}

\newtheorem{constr}[thm]{Construction}
\renewcommand{\vec}[1]{\bm{#1}}

\newcommand{\qedd}{\hspace*{\fill}$\Box$\medskip}

\def\deg{\hbox{\rm{deg\,}}}

\renewcommand{\vec}[1]{\bm{#1}}
\def\tr{\hbox{\rm{tr}}}

\def\wt{\hbox{\rm{wt}}}
\def\AI{\hbox{\rm{AI}}}
\def\dis{\hbox{\rm{d}}}
\def\Supp{\hbox{\rm{supp}}}
\def\Sup{{\rm{supp}}}

\begin{document}

\title{Constructing $2m$-variable Boolean functions with optimal algebraic
immunity based on polar decomposition of
$\mathbb{F}_{2^{2m}}^*$\thanks{Partially supported by National Basic
Research Program of China (2011CB302400).}}

\author{Jia Zheng\thanks{School of Mathematical Sciences, University of Chinese Academy of Sciences, Beijing 100049, China.
Email: zhengjia11b@mail.ucas.ac.cn},
 Baofeng Wu\thanks{Key
Laboratory of Mathematics Mechanization, AMSS, Chinese Academy of
Sciences,
 Beijing 100190,  China. Email: wubaofeng@amss.ac.cn},
 Yufu Chen\thanks{School of Mathematical Sciences, University of Chinese Academy of Sciences, Beijing 100049, China. Email:
 yfchen@ucas.ac.cn},
 Zhuojun Liu\thanks{Key Laboratory of Mathematics
Mechanization, AMSS, Chinese Academy of Sciences,
 Beijing 100190,  China. Email: zliu@mmrc.iss.ac.cn}}
 \date{}

\maketitle

\begin{abstract}
Constructing $2m$-variable Boolean functions with optimal algebraic
immunity based on decomposition of additive group of the finite
field $\mathbb{F}_{2^{2m}}$ seems to be a promising approach since
Tu and Deng's work. In  this paper, we consider the same problem in
a new way. Based on polar decomposition of the multiplicative group
of $\mathbb{F}_{2^{2m}}$, we propose a new construction of Boolean
functions with optimal algebraic immunity. By a slight modification
of it, we obtain a class of balanced Boolean functions achieving
optimal algebraic immunity, which also have optimal algebraic degree
and high nonlinearity. Computer investigations imply that this class
of functions also behave well against fast algebraic attacks.
\vskip .5em


\end{abstract}


\section{Introduction}
\label{intro}

Boolean functions play an important role in symmetric cryptography,
especially in the stream ciphers based on linear feedback shift
resisters (LFSRs). They can be used as building blocks in such key
stream generators as filter generator and combiner generator. Due to
the existence of different kinds of known attacks to stream ciphers,
Boolean functions that are useable should satisfy some main criteria
such as balancedness, high algebraic degree, high nonlinearity and
optimal algebraic immunity.

The notion of algebraic immunity was introduced in \cite{AI04} by
Meier et al. after the great success of algebraic attacks to such
well-known stream ciphers as Toyocrypt and LILI-128 \cite{AA03}. In
fact, the algebraic immunity of a Boolean function $f$ is the
smallest possible degree of the nonzero Boolean functions that can
annihilate $f$ or $f+1$. If it is not big enough, the multivariate
polynomial systems derived from the stream ciphers involving $f$
would be efficiently solved, and hence the secret key can be
recovered. This is just the clever idea of the standard algebraic
attacks introduced (improved, more definitely) by Courtois and Meier
\cite{AA03}. It can be proved that the best possible value of the
algebraic immunity of $n$-variable Boolean functions is
$\lceil\frac{n}{2}\rceil$ \cite{AA03}, thus functions attaining this
upper bound are often known as algebraic immunity optimal functions,
or OAI functions for short.

After OAI Boolean functions were introduced, the natural question of
constructing them was considered in a series of work (see e.g.
\cite{CDKS06,DMSBasic06,NQ06,NLQ08}). But the initial constructions
only focused on the criterion of optimal algebraic immunity and did
not satisfy other criteria of Boolean functions, so they were just
of more interest in theory. Besides, though having optimal algebraic
immunity, these functions did not resist fast algebraic attacks
(FAAs) well. The technique of fast algebraic attack is improved from
the standard algebraic attack, the key point of which is to find low
degree multiples of Boolean functions used in the ciphers to be
attacked such that their products are of reasonable degree
\cite{FAA03}. No progress in constructing Boolean functions having
all ``good" properties was made until 2008. In their pioneering
work, Carlet and Feng proposed an infinite class of balanced Boolean
functions which had optimal algebraic immunity, optimal algebraic
degree and high nonlinearity \cite{CF08}. Computer experiments
implied that the constructed functions also behaved well against
fast algebraic attacks (in fact, very recently this was validated by
Liu et al. in theory \cite{Mliu12}).

In fact, Carlet and Feng seem to have suggested a principle of
constructing Boolean functions achieving optimal algebraic immunity
from finite fields, that is consecutive powers of primitive elements
of certain cyclic groups should be involved in the functions'
supports, which can promise the utility of BCH bound from coding
theory in proving the optimal algebraic immunity of the constructed
functions. Following this principle, Tu and Deng tried a new idea
and (almost) succeeded. They constructed a class of $2m$-variable
Boolean functions based on the additive decomposition
\begin{equation}\label{addecom}
\mathbb{F}_{2^{2m}}=\mathbb{F}_{2^m}\times\mathbb{F}_{2^m}
\end{equation}
which optimized most of the criteria \cite{TD11}, but had two
drawbacks that the optimal algebraic immunity of them could only be
proved assuming the correctness of a combinatorial conjecture, and
the ability of them resisting fast algebraic attacks is bad
\cite{carlet2009}. Afterwards, Tang et al. adopted a similar
technique, constructing a class of OAI functions which also had
other good properties and good immunity against fast algebraic
attacks \cite{TDT13} (in fact, this was stated by Tang et al. based
on computer experiments firstly and proved by Liu et al. in theory
lately  \cite{mliu12biv}). Very recently, Jin et al. found a general
 construction that could involve Tu and Deng's construction and Tang
 et al.'s construction as special cases \cite{Jin11}. The optimal algebraic
 immunity of these functions was proved based on a general conjecture
 proposed in \cite{TDT13}. In all these constructions of even
 variable OAI functions, the ``certain cyclic group" was chosen to
 be $\mathbb{F}_{2^m}^*$, the multiplicative group of the finite field $\mathbb{F}_{2^m}$.

In addition to the additive decomposition \eqref{addecom} of
$\mathbb{F}_{2^{2m}}$, we also have a multiplicative decomposition
of $\mathbb{F}_{2^{2m}}^*$ like
\begin{equation}\label{multdecom}
\mathbb{F}_{2^{2m}}^*=\mathbb{F}_{2^m}^*\times U,
\end{equation}
where $U$ is a cyclic subgroup of $\mathbb{F}_{2^{2m}}^*$ of order
$(2^m+1)$. In fact, instead of multiplicative decomposition, this
decomposition is often known as the polar decomposition of
$\mathbb{F}_{2^{2m}}^*$, which can be used to construct bent and
hyper-bent functions \cite{mesnager2011}, and vectorial Boolean
functions achieving high algebraic immunity \cite{YHH12}. By
choosing the ``certain cyclic group" in Carlet and Feng's principle
to be $\mathbb{F}_{2^m}^*$, we propose a new construction of
$2m$-variable OAI Boolean functions based on the polar decomposition
\eqref{multdecom} in this paper, which can be viewed as a
multiplicative analog of Tu and Deng's construction. After modifying
these functions to be balanced ones, we obtain Boolean functions
satisfying almost all main criteria and potentially behaving well
against fast algebraic attacks (by potentially we mean that this is
only supported by computational evidence up to present). A big
difference in the ``modifying to be balanced" process between our
construction and the former ones is that, something should be
subtracted from the supports of the initially constructed functions
since they are ``fatter" than that of balanced functions in our
construction, while something should be added to the supports of the
initially constructed functions since they are ``thinner" than that
of balanced functions in the former constructions.

The rest of the paper is organized as follows. In Section 2, we give
the necessary preliminaries concerning Boolean functions. In Section
3, we prove a useful combination result, based on which we construct
a class of OAI Boolean functions in Section 4. In Section 5, these
functions are modified to be balanced ones which are also OAI
functions, and their algebraic degree, nonlinearity and  behavior
resisting fast algebraic attacks are studied. Concluding remarks are
given in Section 6.


\section{Preliminary}
\label{sec:1} Let $\mathbb{F}_{2}$  be the binary finite field and
$\mathbb{F}_{2}^{n}$ be the $n$-dimensional vector space over
$\mathbb{F}_{2}$. An $n$-variable Boolean function is a mapping from
$\mathbb{F}_{2}^{n}$ to $\mathbb{F}_{2}$. Denote by $\mathbb{B}_{n}$
the set of all $n$-variable Boolean functions. The support of a
Boolean function $f$ is defined as $$\Supp(f)=\{\vec x \in
\mathbb{F}_{2}^{n}\mid f(\vec x)=1\},$$ and the cardinality of it,
$\wt(f)$, is called the Hamming weight of $f$. Furthermore, for
another Boolean function $g\in\mathbb{B}_n$, the distance between
$f$ and $g$ is defined as $\dis(f,g)=\wt(f+g)$. When
$\wt(f)=2^{n-1}$, we call $f$ a balanced function. Abusing
notations, we also denote the Hamming weight of a vector $\vec
v\in\mathbb{F}_{2}^{n}$, i.e. the number of  nonzero positions of
$\vec v$, to be $\wt(\vec v)$. Besides, for an integer $u$, we
denote by $\wt_n(u)$ the number of 1's in the binary expansion of
the reduction of $u$ modulo $(2^n-1)$ in the complete residue system
$\{0,1,\ldots,2^n-2\}$. Obviously, $\wt_n(-u)=n-\wt_n(u)$ when
$2^{n}-1\,\nmid\,u$.

There are several ways to describe a Boolean function such as by its
truth table, algebraic normal form (ANF), univariate representation
and so on. Each $f \in \mathbb{B}_{n}$ has a unique ANF of the form
$$f(x_{1},\ldots,x_{n})=\sum_{I\subseteq \{1,2,\ldots,n\}}a_{I}
\,\prod_{i\in I}x_{i},~~a_I\in\mathbb{F}_{2}. $$ The algebraic
degree of $f$, $\deg(f)$, is defined to be $\max\{ |I| \mid
a_{I}\neq 0 \}$. It should be noted that for  $n$-variable balanced
Boolean functions, the maximal possible algebraic degree is $(n-1)$.
Boolean functions of degree at most $1$ are called affine functions,
and the set of all of them are denoted to be $\mathbb{A}_{n}$.
 In order to resist the fast correlation attacks, Boolean functions
used in  cryptographic systems should have high nonlinearity, where
the nonlinearity of a Boolean function $f$, $\mathcal{N}_f$, is
defined as the minimum distance between $f$ and all affine
functions, i.e.
$$\mathcal{N}_{f}=\min_{a\in \mathbb{A}_{n}}\dis(f,a).$$

Walsh transform  is a powerful tool in studying Boolean functions.
For any $\vec\lambda\in\mathbb{F}_{2}^{n}$, the Walsh transform of
$f\in \mathbb{B}_{n}$ at $\vec\lambda$ is defined by
 $$W_{f}(\vec\lambda)=\sum_{\vec x\in\mathbb{F}_{2}^{n}}(-1)^{f(\vec x)+\vec\lambda \cdot \vec x},$$
 where ``$\cdot$" represents  the Euclidean  inner product of vectors.
Many criteria of $f$ can be described by its Walsh transform such as
balancedness, nonlinearity and correlation immunity
\cite{carletBFbook}. For example, we have $W_f(0)=0$ when $f$ is
balanced, and we can equivalently express nonlinearity of $f$ by
$$\mathcal{N}_{f}=2^{n-1}-\max_{\vec\lambda\in \mathbb{F}_{2}^{n}}\mid
W_{f}(\vec\lambda)|.$$

As is well known that the finite field $\mathbb{F}_{2^{n}}$ is
isomorphic to $\mathbb{F}_{2}^{n}$ through the choice of a basis of
$\mathbb{F}_{2^{n}}$ over $\mathbb{F}_{2}$, hence naturally, the
Boolean function $f$ can be represented by a univariate polynomial
over $\mathbb{F}_{2^{n}}$ of the form
$$f(x)=\sum_{i=0}^{2^{n}-1}f_{i}x^{i}.$$
It can be proved that as a Boolean function, the coefficients of $f$
satisfy $f_{2i}=f_{i}^{2}$ (subscripts reduced modulo $(2^n-1)$) for
$1\leq i\leq 2^n-2$ and $f_{0},~f_{2^{n}-1}\in \mathbb{F}_{2}$.
Besides, it is not difficult to deduce that
$$\deg(f)=\max\{\wt_n(i)\mid f_{i} \neq 0,~ 0\leq i\leq 2^{n}-1\}.$$
Under univariate representation, the Walsh transform of $f$ at
$\lambda\in\mathbb{F}_{2^n}$ can be described as
$$W_{f}(\lambda)=\sum_{x\in\mathbb{F}_{2^{n}}}(-1)^{f(x)+\tr_{1}^{n}(\lambda x)},$$
where $\tr_{1}^{n}(\cdot)$ is the trace function from
 $\mathbb{F}_{2^{n}}$ to $\mathbb{F}_{2}$, i.e.
 $\tr_{1}^{n}(x)=\sum_{i=0}^{n-1}x^{2^i}$ for any $x\in\mathbb{F}_{2^{n}}$.

When $n$ is even, we can give another formulation of the univariate
representation of Boolean functions based on  polar decomposition of
$\mathbb{F}_{2^{n}}^{*}$. Let $n=2m$. Then $\mathbb{F}_{2^{m}}^{*}$
is a cyclic subgroup of $\mathbb{F}_{2^{n}}^{*}$. Since
$(2^m-1,\frac{2^n-1}{2^m-1})=(2^m-1,2^m+1)=1$, there exists a cyclic
subgroup $U$ of $\mathbb{F}_{2^{n}}^{*}$ of order $2^m+1$ such that
\[\mathbb{F}_{2^{n}}^{*}=\mathbb{F}_{2^{m}}^{*}\times U .\]
This is just the polar decomposition of $\mathbb{F}_{2^{n}}^{*}$. If
we assume $\alpha$ to be a primitive element of
$\mathbb{F}_{2^{n}}$, then it is obvious that $U=\langle\xi\rangle$
where $\xi=\alpha^{2^{m}-1}$.  From the polar decomposition we know
that any $x\in\mathbb{F}_{2^{n}}^{*}$ can be represented as $x=yz$
for some $y\in\mathbb{F}_{2^{m}}^{*}$ and $z\in U$. Then we can
represent the Boolean function $f$ by
\[f(x)=
\left\{
\begin{array}
    {c@{~\text{if}~}l}
f_{0} &  x=0; \\
f'(x)=f'(y,z)& 0\neq x=yz, ~ y\in \mathbb{F}^{*}_{2^{m}},~z\in U,
\end{array}
\right.\] where $ f'(x)=\sum_{i=0}^{2^{n}-2} f^{'}_{i}x^{i}$ is the
polynomial representation of the map
$\mathbb{F}_{2^{n}}^{*}\longrightarrow\mathbb{F}_{2}$, $c\longmapsto
f(c)$ (by Lagrange interpolation). Note that
$$
f'(y,z)=\sum_{i=0}^{2^{n}-2} f'_{i}(yz)^{i}=\sum_{j=0}^{2^{m}-2}
\sum_{k=0}^{2^{m}}f'_{j,k}y^{j}z^{k}
$$ where for any $0\leq i\leq 2^n-2$, $f'_i=f'_{jk}$ if and only
if $\left\{\begin{array}{c}i\equiv j\mod(2^m-1)\\i\equiv k\mod
(2^m+1)
\end{array}\right.$, i.e. $i\equiv 2^{m-1}((2^{m}+1)j+(2^{m}-1)k)\mod(2^n-1)$ (by the Chinese
remainder theorem). Besides,
\begin{eqnarray*}
f(x)&=&f_{0}(x^{2^{n}-1}+1)+ f'(x)x^{2^{n}-1}\\
    &=&f_{0}+f_{0}x^{2^{n}-1}+x^{2^{n}-1}\sum_{i=0}^{2^{n}-2}
    f'_{i}x^{i}\\
    &\equiv&f_{0}+(f_{0}+f'_{0})x^{2^{n}-1}+\sum_{i=1}^{2^{n}-2}
    f'_{i}x^{i} \mod (x^{2^n}+x),
\end{eqnarray*}
hence the algebraic degree of $f$ can be expressed as
\[\deg(f)=
\left\{
\begin{array}
    {c@{~\text{if}~}l}
     \max\{\wt_{n}(2^{m-1}((2^{m}+1)j+(2^{m}-1)k))\mid f'_{j,k}\neq 0\} & f_{0}+f'_{0}=0; \\
      n & f_{0}+f'_{0}\neq 0. \\
\end{array}
\right.
\]
That is to say, if the  algebraic degree of $f$ is smaller than $n$,
we have $f_{0}=f'_{0}=f'_{0,0}$ and
\begin{eqnarray*}
 \deg(f)  &=& \max\{\wt_{n}(2^{m-1}((2^{m}+1)j+(2^{m}-1)k))\mid f'_{j,k}\neq 0\} \\
   &=&\max\{\wt_{n}((2^{m}+1)j+(2^{m}-1)k)\mid f'_{j,k}\neq 0\}.
\end{eqnarray*}

To finish this section, we recall the definition of algebraic
immunity of Boolean functions.
\begin{defi}
Let $f,g \in \mathbb{B}_{n}$. $g$ is called an annihilator of $f$ if
$fg=0$. The algebraic immunity  of $f$, $\AI(f)$, is defined to be
the smallest possible degree of the nonzero annihilators of $f$ or
$f+1$, i.e.
$$\AI(f)=\min_{0\neq g \in \mathbb{B}_{n}}\{\deg(g)\mid fg=0 \text{ or
}(f+1)g=0\}.$$
\end{defi}


\section{A combination fact} \label{sec:2}

In this section, we prove a useful combination result about the
weight distribution of integers, which will be of key importance in
proving the optimal algebraic immunity of the Boolean functions
constructed in the following sections.

\begin{lem}\label{P1}
Let $n=2m$. Then for any$0 \leq j \leq 2^{m}-2 $, $1 \leq k \leq
2^{m} $, we have
$$\wt_{n}((2^{m}+1)(2^{m}-1-j)+(2^{m}-1)k)=n-\wt_{n}((2^{m}+1)j+(2^{m}-1)k).$$
\end{lem}
\begin{proof}
Obviously,
\begin{eqnarray*}
2^{m}[2^{m}(j-k)+(j+k)]&\equiv& 2^{n}(j-k)+2^{m}(j+k)   \\
                       &\equiv&2^{m}(j+k)+(j-k)  \mod  (2^{n}-1),
\end{eqnarray*}
and thus
\begin{eqnarray*}
 \wt_n(2^{m}(j-k)+(j+k))  &=& \wt_n(2^{m}[2^{m}(j-k)+(j+k)]) \\
   &=& \wt_n(2^{m}(j+k)+(j-k)) \\
   &=&\wt_n((2^{m}+1)j+(2^{m}-1)k).
\end{eqnarray*}
Then we get
\begin{eqnarray*}
          wt_{n}((2^{m}+1)(2^{m}-1-j)+(2^{m}-1)k)
           &=&wt_{n}(2^n-1-(2^{m}+1)j+(2^{m}-1)k)\\
           &=&n-wt_{n}((2^{m}+1)j-(2^{m}-1)k)\\
           &=&n-wt_{n}(2^{m}(j-k)+(j+k))\\
           &=&n-\wt_n((2^{m}+1)j+(2^{m}-1)k).
\end{eqnarray*}
\qedd
\end{proof}

\begin{prop}\label{P3}
Let $n=2m$. For any $0\leq k\leq 2^m$, define
$$S_{k}=\{j\in\mathbb{Z}/(2^{m}-1)\mathbb{Z}\mid\wt_{n}((2^{m}+1)j+(2^{m}-1)k)<m\}.$$
Then $|S_{k}| \leq 2^{m-1}$, $0\leq k \leq 2^{m}$. Moreover, ``$=$"
holds if and only if $m$ is odd and $k=0$.
\end{prop}

\begin{proof}
Consider  the case $k=0$ firstly. Since
\begin{eqnarray*}
S_{0}&=&\{j\in\mathbb{Z}/(2^{m}-1)\mathbb{Z}\mid\wt_{n}((2^{m}+1)j)<m\}\\
     &=&\{j\in\mathbb{Z}/(2^{m}-1)\mathbb{Z}\mid\wt_{n}(2^{m}j+j)<m\}\\
     &=&\{j\in\mathbb{Z}/(2^{m}-1)\mathbb{Z}\mid\wt_{m}(j)<\frac{m}{2}\},
\end{eqnarray*}
it is easy to get
\begin{eqnarray*}
|S_{0}|&=&\left\{
          \begin{array}{c@{~\text{if}~}l}
            \sum^{\lfloor\frac{m}{2}\rfloor}_{i=0}\binom{m}{i}&  m\text{ is odd};\\
            \sum^{\frac{m}{2}-1}_{i=0}(^{m}_{i})&  m\text{ is even}
          \end{array}
        \right.\\
        &=&\left\{
          \begin{array}{c@{~\text{if}~}l}
            2^{m-1}&  m\text{ is odd};\\
            2^{m-1}-\frac{1}{2}\binom{m}{m/2}&  m\text{ is even},
          \end{array}
        \right.
\end{eqnarray*}
which implies that $|S_{0}|<2^{m-1}$ when $m$ is even.

Now we consider the case $1 \leq k \leq 2^{m}$. Define the set
$$T_{k}=\{j\in\mathbb{Z}/(2^{m}-1)\mathbb{Z}\mid\wt_{n}((2^{m}+1)j+(2^{m}-1)k)>m\}.$$
From Lemma \ref{P1} we have for any $0\leq j\leq 2^m-2$,
$$\wt_{n}((2^{m}+1)(2^{m}-1-j)+(2^{m}-1)k)=n-\wt_{n}((2^{m}+1)j+(2^{m}-1)k),$$
thus $|S_{k}|=|T_{k}|$. On the other hand, since
\begin{eqnarray*}
    \wt_{n}((2^{m}-1)k)&=&\wt_{n}(2^{m}(2^{m}k-k))\\
    &=&\wt_{n}(k-2^{m}k)\\
    &=&n-\wt_{n}((2^{m}-1)k),
\end{eqnarray*}
i.e. $wt_{n}((2^{m}-1)k)=\frac{n}{2}=m$, we know that $0 \not\in
S_{k}$ and $0\not\in T_{k}$, which implies
$$|S_{k}|+|T_{k}|\leq
2^{m}-2 $$ as $S_k\cap T_k=\emptyset$. Then it follows that
$|S_k|\leq 2^{m-1}-1<2^{m-1}$. \qedd
\end{proof}


\section{A class of unbalanced OAI Boolean functions}\label{sec:3}

In this section, based on polar decomposition of
$\mathbb{F}_{2^n}^*$ and the combination results in Section
\ref{sec:2}, we construct a new class of OAI Boolean functions.

\begin{constr}\label{con1}
Let $n=2m$. Let $\beta$ be a primitive element of $\mathbb{F}_{2^m}$
and $U$ be the cyclic group defined in Section \ref{sec:1}. Set
$\Delta=\{1,\beta,\beta^{2},\ldots,\beta^{2^{m-1}-1}\}$. Define an
$n$-variable Boolean function $f$ by setting
$$ \Supp(f)= \Delta\times U.$$
\end{constr}

\begin{thm}\label{T1}
Let $f$ be the  Boolean function defined in Construction \ref{con1}.
Then $f$ has optimal algebraic immunity.
\end{thm}

\begin{proof}
From the definition of algebraic immunity, it suffices to prove that
there is no nonzero annihilator with degree smaller than $m$ of both
$f$ and $f+1$.

Suppose $g \neq 0$ is an annihilator of $f$ with algebraic degree
smaller than $m$. Assume
\[g(x)=
\left\{
\begin{array}
    {c@{~\text{if}~}l}
    g'(y,z) & 0\neq x=yz, ~ y\in \mathbb{F}^{*}_{2^{m}},~z\in U; \\
   g_{0}  & x=0,
\end{array}
\right.
\]
where
$\displaystyle{g'(y,z)=\sum_{j=0}^{2^{m}-2}\sum_{k=0}^{2^{m}}g_{j,k}y^{j}z^{k}}$,
$g_{j,k}\in\mathbb{F}_{2^{n}}$, $ g_{0} \in \mathbb{F}_{2}$. Then
$$g'(y,z)=\sum_{k=0}^{2^{m}}\left(\sum_{j=0}^{2^{m}-2}g_{j,k}y^{j}\right)z^{k}=\sum_{k=0}^{2^{m}}g_{k}(y)z^{k}=0$$
for all $z\in U$ and $y\in\Delta$, where
$\displaystyle{g_{k}(y)=\sum_{j=0}^{2^{m}-2}g_{j,k}y^{j}}$. For any
fixed $y_{0}\in \Delta$, since $g'(y_{0},z)$ has (\,$2^{m}+1$\,)
zeros, we conclude that $g_{k}(y)=0$ for any $y \in \Delta$, $0 \leq
k \leq 2^{m}$.

On the one hand, from the definition of BCH code \cite{MacError77},
we know that for $0\leq k\leq 2^m$,
$(g_{0,k},g_{1,k},g_{2,k},\ldots,g_{2^{m}-2,k})$ is a codeword of
some BCH code over $\mathbb{F}_{2^{n}}$ of length $(2^{m}-1)$  with
elements in $\Delta$ as zeros. Thus based on the BCH bound, the
Hamming weight of a nonzero codeword should be greater than or equal
to $(2^{m-1}+1)$, i.e.
$$\wt(g_{0,k},g_{1,k},g_{2,k},\ldots,g_{2^{m}-2,k})\geq 2^{m-1}+1.$$
On the other hand, since $\deg(g)<m$, we have $g_{j,k}=0$ if
$\wt_{n}((2^{m}+1)j+(2^{m}-1)k)\geq m $. Form Proposition \ref{P3},
we know  that $|S_{k}|\leq2^{m-1}$. That is
   \begin{eqnarray*}
       \wt(g_{0,k},g_{1,k},g_{2,k},\cdots,g_{2^{m}-2,k})\leq
        2^{m-1},
    \end{eqnarray*}
which leads to a contradiction. Hence $g=0$.

Next, we consider the function $f+1$. Note that
$$  \Supp(f+1) =\Delta'\times U \cup \{0\} $$
 where
 $\Delta'=\{\beta^{2^{m-1}},\beta^{2^{m-1}+1},\ldots,\beta^{2^{m}-2}\}$.
Similar to the proof with respect to $f$, we let $g$ be now a
nonzero annihilator of $f+1$ with algebraic degree smaller than $m$.
We can deduce from the BCH bound  that, for any $0\leq k\leq 2^m$,
the vector $(g_{0,k},g_{1,k},g_{2,k},\cdots,g_{2^{m}-2,k})$ has
weight at least $2^{m-1}$ since $|\Delta'|=2^{m-1}-1$. By
Proposition \ref{P3}, we know that when $m$ is even, the weight of
$(\,g_{0,k},g_{1,k},\cdots,g_{2^{m}-2,k}\,)$ is smaller than
$2^{m-1}$, thus a contradiction follows and $g=0$. When  $m$ is odd,
we have
$(g_{0,k},g_{1,k},g_{2,k},\cdots,g_{2^{m}-2,k})=(0,0,\ldots,0)$ for
$1\leq k\leq 2^m$. Since $|S_0|=2^{m-1}$, we get
$\wt((g_{0,0},g_{1,0},\cdots,g_{2^{m}-2,0}))=2^{m-1}$, which implies
that $g_{0,0}=g_0=1$.
 However, this contradicts the fact that
$0\in\Supp(f+1)$. We also have $g=0$.

To summarize, we know that $f$ has optimal algebraic immunity. \qedd
\end{proof}

\begin{rmk}
From the proof of Theorem \ref{T1}, it is easy to see that if we
replace the set $\Delta$ in Construction \ref{con1} by
$\{\beta^s,\beta^{s+1},\ldots,\beta^{s+2^{m-1}-1}\}$ for any $0\leq
s\leq 2^m-2$, we can also obtain Boolean functions with optimal
algebraic immunity.
\end{rmk}

It is direct to find that the weight of the function in Construction
\ref{con1} is $(2^{n-1}+2^{m-1})$, which is bigger than that of
balanced functions. Thus we do not talk about their further
properties since they are not of applicable interest.


\section{Balanced functions with optimal algebraic immunity and other good
properties}\label{sec:4}

In this section, we modify the functions in Construction \ref{con1}
to be balanced ones which maintain optimal algebraic immunity by
changing some points between their supports and zeros. Furthermore,
we study in detail properties of these balanced functions such as
their algebraic degree, nonlinearity and immunity against fast
algebraic attacks.

\begin{constr}\label{con2}
Let $n=2m$. Let $\alpha$ be a primitive element of
$\mathbb{F}_{2^n}$ and $\beta=\alpha^{2^m+1}$, $\xi=\alpha^{2^m-1}$.
Set $\Gamma=\{\beta,\beta^{2},\ldots,\beta^{2^{m-1}-1}\}$. Define an
$n$-variable Boolean function $F$ by setting
$$ \Supp(F)= (\Gamma\times U) \cup (\{1\}\times\{1,\xi,\ldots,\xi^{2^{m-1}}\}).$$
\end{constr}

\begin{thm}\label{T2}
Let $F$ be the  Boolean function defined in Construction \ref{con2}.
Then $F$ is balanced and has optimal algebraic immunity.
\end{thm}

\begin{proof}
It is obvious that
$\wt(F)=(2^{m-1}-1)\times(2^{m}+1)+2^{m-1}+1=2^{n-1}$, so $F$ is
balanced.

The proof of optimal algebraic immunity of $F$ is similar to that of
Theorem \ref{T1}. Suppose $g$ is a nonzero annihilator of $F$ with
algebraic degree smaller than $m$, and assume
\[g(x)= \left\{
\begin{array}
    {c@{~\text{if}~}l}
    g'(y,z) & 0\neq x=yz,~ y\in \mathbb{F}^{*}_{2^{m}},~z\in U; \\
   g_{0}  & x=0,
\end{array}
\right.
\]
where
$\displaystyle{g'(y,z)=\sum_{j=0}^{2^{m}-2}\sum_{k=0}^{2^{m}}g_{j,k}y^{j}z^{k}}$,
$g_{j,k}\in\mathbb{F}_{2^{n}}$, $ g_{0} \in \mathbb{F}_{2}$. Since
$\{\beta,\beta^{2},\ldots,\beta^{2^{m-1}-1}\} \times U \subseteq
\Supp(f)$, by the BCH bound and Proposition \ref{P3}, we get for
$k>0$, $(g_{0,k},g_{1,k},\ldots,g_{2^{m}-2,k})=(0,0,\ldots,0)$. Then
$g'(y,z)$ turns to $g'(y,z)=\sum_{j=0}^{2^{m}-2}g_{j,0}y^{j}$.
Besides, as $\{1\}\times\{1,\xi,\ldots,\xi^{2^{m-1}}\} \subseteq
\Supp(f)$, we have
$$g'(1,z)=\sum_{j=0}^{2^{m}-2}g_{j,0}1^{j}=0,$$ which means that
$\{1,\beta,\beta^{2},\ldots,\beta^{2^{m-1}-1}\}$ are  zeros of
certain BCH code containing $(g_{0,0},g_{1,0},\ldots,g_{2^{m}-2,0})$
as a codeword. Using the BCH bound and Proposition \ref{P3} again,
we obtain a contradiction. Thus $F$ has no nonzero annihilator with
degree smaller than $m$. With respect to $F+1$, the proof procedure
is almost the same.

Finally,  we conclude that the Boolean function $F$ has optimal
algebraic immunity.\qedd
\end{proof}

\begin{rmk}
From the proof of Theorem \ref{T2}, it is not difficult to see that
we can also set $ \Supp(F)=
(\{1,\beta,\ldots,\beta^{2^{m-1}-2}\}\times U) \cup
(\{\beta^{2^{m-1}-1}\}\times\{1,\xi,\ldots,\xi^{2^{m-1}}\})$ to
obtain balanced Boolean functions with optimal algebraic immunity.
\end{rmk}

\subsection{Polynomial representation and algebraic degree}

In the following, we compute  the univariate representation of the
OAI Boolean function $F$ in Construction \ref{con2} and deduce its
algebraic degree.

By the Chinese remainder theorem, we can write the support of $F$ in
the  form
\begin{eqnarray*}
 \Supp(F)  &=& \{\alpha^{2^{m-1}((2^m+1)j+(2^m-1)k)}\mid 1 \leq j \leq 2^{m-1}-1,~0\leq k \leq 2^{m}\} \\
   &&\cup\{ \alpha^{2^{m-1}(2^m-1)k}\mid 0 \leq k \leq 2^{m-1}\}.
\end{eqnarray*}

For simplicity, we distinguish the integer
$2^{m-1}((2^m+1)j+(2^m-1)k)$ reduced modulo $(2^n-1)$ with a pair
$(j,k)$ where $0\leq j\leq 2^m-2$, $0\leq k\leq 2^m$. It is easy to
find that
\begin{eqnarray*}
 (j+1,k+1) &=& 2^{m-1}((2^{m}+1)(j+1)+(2^{m}-1)(k+1)) \\
           &=& 2^{m-1}((2^{m}+1)j+(2^{m}-1)k)+1\\
           &=& (j,k)+1,
\end{eqnarray*}
and
\begin{eqnarray*}
  (j,k-2)&=& 2^{m-1}((2^{m}+1)j+(2^{m}-1)(k-2)) \\
         &=& 2^{m-1}((2^{m}+1)j+(2^{m}-1)k)+(2^{m}-1)\\
         &=& (j,k)+(2^{m}-1).
\end{eqnarray*}
Using these properties, we can derive that the support of $F$ is
just
\begin{eqnarray*}
  \Supp(F) &=&\{\alpha^{(j,k)}\mid 1 \leq j \leq 2^{m-1}-1,~0\leq k \leq 2^{m}\}
\cup \{ \alpha^{(0,k)}\mid 0 \leq k \leq 2^{m-1}\}  \\
   &=&\{\alpha^{l(2^{m}-1)+r}\mid 0\leq l \leq 2^{m},~1 \leq r
   \leq2^{m-1}-1\}\\
   &&\cup\{\alpha^{2^{m-1}(2^{m}-1)k} \mid 0\leq k \leq2^{m-1}\}.
\end{eqnarray*}
Then the coefficients of the function $f'$ whose support is the
first part of $\Supp(F)$ can be decribed explicitly, i.e. for
$0<i<2^{n}-1$,
\begin{eqnarray*}
 f'_{i}&=&\sum_{l=0}^{2^{m}}\sum_{j=l(2^{m}-1)+1}^{l(2^{m}-1)+2^{m-1}-1}(\alpha^{-i})^{j}\\
      &=&\sum_{l=0}^{2^{m}}\frac{(\alpha^{-i})^{1+l(2^{m}-1)}(1-(\alpha^{-i})^{2^{m-1}-1})}{1-\alpha^{-i}}\\
      &=&\dfrac{\alpha^{-i}(1-\alpha^{-i(2^{m-1}-1)})}{1-\alpha^{-i}}\sum_{l=0}^{2^{m}}\alpha^{-il(2^{m}-1)}\\
      &=&\left\{
          \begin{array} {c@{~\text{if}~}l}
          0 & 2^{m}+1\nmid i  ; \\
            \dfrac{\alpha^{-i}(1-\alpha^{-i(2^{m-1}-1)})}{1-\alpha^{-i}} & 2^{m}+1\mid i.
          \end{array}
        \right.
\end{eqnarray*}
Similarly,  the coefficients of $f''$ whose support is the second
part of $\Supp(F)$ are that, for $0<i<2^{n}-1$,
\begin{eqnarray*}
 f''_{i}&=&\sum_{k=0}^{2^{m-1}}(\alpha^{-i})^{2^{m-1}(2^{m}-1)k}\\
        &=&\left\{
          \begin{array} {c@{~\text{if}~}l}
            \dfrac{1-\alpha^{-i2^{m-1}(2^{m}-1)(2^{m-1}+1)}}{1-\alpha^{-i2^{m-1}(2^{m}-1)}} & 2^{m}+1\nmid i; \\
            1 & 2^{m}+1 \mid i.
          \end{array}
        \right.
\end{eqnarray*}
It is obvious that, if we assume $F(x)=\sum_{i=0}^{2^{n}-1}
F_{i}x^{i}$, then $F_i=f'_{i}+ f''_{i}$ for $1\leq i\leq 2^{n}-2$,
$F_0=0$ (since $F(0)=0$) and $F_{2^n-1}=0$ (since $F$ is balanced).
Hence we can give the univariate representation of $F$.
\begin{thm}\label{T3}
Let $F$ be the $n$-variable Boolean function defined in Construction
\ref{con2}. Then the univariate representation of $F$ is
 $$F(x)=\sum_{i=1}^{2^{n}-2} F_{i}x^{i},$$
where
\[
F_{i} =\left\{
          \begin{array}{c@{~\text{if}~}l}
            \dfrac{1-\alpha^{-i2^{m-1}(2^{m}-1)(2^{m-1}+1)}}{1-\alpha^{-i2^{m-1}(2^{m}-1)}}& 2^{m}+1\nmid i; \\
            1+\dfrac{\alpha^{-i}(1-\alpha^{-i(2^{m-1}-1)})}{1-\alpha^{-i}}& 2^{m}+1 \mid i.
          \end{array}
        \right.
\]
Hence the algebraic degree of  $F$ is $(n-1)$, which is optimal for
balanced Boolean functions.
\end{thm}

\subsection{Nonlinearity }
To determine the lower bound of the nonlinearity of the Boolean
functions in Construction \ref{con2}, we need some necessary
backgrounds.

\begin{defi}[\cite{FiniteField97}]
Let $a \in \mathbb{F}_{2^{m}}$. The binary complete Kloosterman sum
is defined as
\[\mathcal {K}(a)=\sum_{x \in \mathbb{F}_{2^{m}}}(-1)^{\tr_{1}^{m}(1/x+a x)}.\]
\end{defi}

\begin{lem}[\cite{Mbent}]\label{lem2}
Let $a \in \mathbb{F}_{2^{m}}^{*}$ and $U$ be the cyclic group
defined in Section \ref{sec:1}. Then
$$\sum_{z\in U}(-1)^{\tr_{1}^{n}(az)}=1-\mathcal {K}(a).$$
\end{lem}

\begin{lem}[\cite{TDT13}]\label{lem3}
Let $\beta$ be a primitive element of $ \mathbb{F}_{2^{m}}$. Let
$$\Delta_{s}=\{\beta^{s},\beta^{s+1},\ldots,\beta^{2^{m-1}+s-1}\}$$
where $0\leq s< 2^{m}-1$ is an integer. Then
$$\left|\sum_{\gamma\in\Delta_{s}}\left(\mathcal{K}(\gamma)-1\right)\right|< (\frac{\ln2}{\pi}+0.42)2^{m}+1.$$
\end{lem}

\begin{thm}\label{T4}
Let $F$ be the Boolean function defined in Construction \ref{con2}.
Then
$$\mathcal{N}_{F}> 2^{n-1}-(\frac{\ln2}{\pi}m+0.92)2^{m}-1.$$
\end{thm}

\begin{proof}
We denote the set $\{1,\xi,\ldots,\xi^{2^{m-1}}\} $ by $\Lambda$.
Obviously, $W_{F}(0)=0$ since $F$ is balanced.

For any $a \in \mathbb{F}_{2^{n}}^{*}$, we assume $a=a_{1}a_{2}$
where $a_{1} \in \mathbb{F}_{2^{m}}^{*}$, $a_{2}\in U$. By Lemma
\ref{lem2}, we have
\begin{eqnarray*}
 W_{F}(a)&=& -2\sum_{x\in \Sup(F)}(-1)^{\tr_{1}^{n}(ax)}\\
         &=& -2\left[\sum_{y\in\Gamma}\sum_{z\in U}(-1)^{\tr_{1}^{n}(a_{1}ya_{2}z)}+\sum_{z \in
         \Lambda}(-1)^{\tr_{1}^{n}(a_{1}a_{2}z)}\right]\\
         &=&-2\left[\sum_{y\in \Gamma'}(1-\mathcal {K}(a_{1}y))-\sum_{z\in
         U}(-1)^{\tr_{1}^{n}(a_{1}z)}+\sum_{z\in\Lambda'}(-1)^{\tr_{1}^{n}(a_{1}z)}\right]\\
         &=&-2\left[\sum_{y\in \Gamma'}(1-\mathcal {K}(a_{1}y))-\sum_{z\in U \setminus
         \Lambda'}(-1)^{\tr_{1}^{n}(a_{1}z)}\right],
\end{eqnarray*}
where $\Lambda'=\{a_2,a_2\xi,\ldots,a_2\xi^{2^{m-1}}\}$, $\Gamma'=\{
1,\beta,\ldots,\beta^{2^{m-1}-1}\}$.

Since $a_{1}\in \mathbb{F}_{2^{m}}^{*}$, it can be represented as
$a_{1}=\beta^{s}$ for some  $0\leq s \leq 2^{m}-2$. Then
$$\sum_{y \in \Gamma'}(\mathcal
{K}(a_{1}y)-1)=\sum_{\gamma\in\Delta_{s}}(\mathcal {K}(\gamma)-1),$$
where
$\Delta_{s}=\{\beta^{s},\beta^{s+1},\ldots,\beta^{s+2^{m-1}-1}\}$.
By Lemma \ref{lem3}, we know that
$$\left|\sum_{y \in \Gamma'}(1-\mathcal
{K}(a_{1}y)\right|<\left(\frac{\ln2}{\pi}m+0.42\right)2^{m}+1.$$
Therefore,
$$|\,W_{F}(a)\,|<2\left[\left(\frac{\ln2}{\pi}m+0.42\right)2^{m}+1+2^{m-1}\right].$$
Finally we get
\begin{eqnarray*}
\mathcal{N}_{F} &=& 2^{n-1}-\frac{1}{2}\max_{a\in\mathbb{F}_{2^n}} |W_{f}(a)| \\
                &>&
                2^{n-1}-\left(\frac{\ln2}{\pi}m+0.42\right)2^{m}-2^{m-1}-1\\
                &=&
                2^{n-1}-\left(\frac{\ln2}{\pi}m+0.92\right)2^{m}-1.
\end{eqnarray*}
\qedd
\end{proof}

In fact, the lower bound in Theorem \ref{T4} is not satisfactory at
all since  we have used the naive estimation
\[\left|\sum_{z\in U \setminus\Lambda'}(-1)^{\tr_{1}^{n}(a_{1}z)}\right|\leq 2^{m-1}\]
in the proof. Hence it is not so safe to say that the function $F$
has good nonlinearity. Nevertheless, for these $n$ we can compute
the exact value of nonlinearity, it appears good.

Denote by $\mathcal {N}_{C-F}$, $\mathcal {N}_{T-C-T}$ and $\mathcal
{N}_{F}$  the exact values of nonlinearity of the Carlet-Feng
functions \cite{CF08}, the Tang-Carlet-Tang functions \cite{TDT13}
and the functions in Construction \ref{con2} respectively. By a
Magma program, we investigate the exact values of  nonlinearity for
small number of variables under the choice of the default primitive
element of $\mathbb{F}_{2^{n}}$ in Magma system. The results are
displayed in Table \ref{table2}. By the comparison, we find our
functions almost play as well as the Carlet-Feng and
Tang-Carlet-Tang functions.

\begin{table}[H]
\caption{Comparison of the exact values of  Nonlinearity with some
known constructions}\label{table2}
\begin{center}
\begin{tabular}{lllll}
\hline\noalign{\smallskip}
$n$ & $\mathcal {N}_{C-F}$ & $\mathcal {N}_{T-C-T}$ & $\mathcal {N}_{F}$ &  $2^{n-1}-2^{\frac{n}{2}-1}$\\
\noalign{\smallskip}\hline\noalign{\smallskip}
4 & 4 & 4 &  4 & 6 \\

6 & 24 & 22 & 22 & 28 \\

8 & 112 & 108 & 108 & 120 \\

10 & 478 & 476  & 474 & 496 \\

12 & 1970 & 1982 & 1976 & 2016 \\

14 & 8036 & 8028 & 8026 & 8128 \\

16 & 32530 & 32508 & 32498 & 32540 \\

18 & 130442 & 130504 & 130484 & 130812 \\

20 & 523154 & 523144 & 523122 & 523776 \\
\noalign{\smallskip}\hline
\end{tabular}
\end{center}
\end{table}

\normalsize

To obtain better estimation of the nonlinearity of the functions in
Construction \ref{con2}, the key difficulty is to estimate such
exponential sums as
\[\Phi_s=\sum_{x\in\{\xi^s,\xi^{s+1},\ldots,\xi^{s+2^{m-1}-1}\}}(-1)^{\tr^n_1(cx)}\]
for any $0\leq s\leq 2^m$ and $c\in\mathbb{F}_{2^m}^*$, where $\xi$
is a generator of the cyclic group $U$. Unfortunately, the standard
technique of using Gauss sums would not work for this kind of
incomplete exponential sums over finite fields. Maybe more advanced
number theoretic tools should be introduced to overcome this
difficulty. Though we have not found them up to present, we
conjecture that $|\Phi_s|=O(2^{\frac{m}{2}})$.

\subsection{Immunity against fast algebraic attacks}

The property of optimal algebraic immunity is a necessary but not
sufficient condition for a Boolean function because of the existence
of fast algebraic attacks. In this subsection, we analyze the
ability of the Boolean functions in Construction \ref{con2} against
 fast algebraic attacks.

An $n$-variable Boolean function $f$ is optimal with respect to fast
algebraic attacks if for any pair of integers $(e,d)$ such that
$e+d<n$ and $e<n/2$, there do not exist a function $g \neq 0$ of
algebraic degree at most $e$ such that $fg$ has degree at most $d$
\cite{FAA03}. Armknecht et.al. proposed an efficient algorithm
\cite{FC06} to determine the existence of $g$ and $h$ with
corresponding degrees. Based on Algorithm 2 in \cite{FC06}, we
investigate the behavior of the functions in Construction \ref{con2}
against fast algebraic attacks for small number of variables by a
Magma program.

We choose the default primitive element of $ \mathbb{F}_{2^{n}}$ in
the Magma system. For even $n$ ranging from 4 to 14 and
$e<\frac{n}{2}$, we can find the pairs $(e,d)$ with $e+d \geq n-1$,
but the pairs $(e,d)$ such that $e+d \leq n-2$ have never been
observed. That implies that the functions in Construction \ref{con2}
have good immunity to fast algebraic attacks though they are not
optimal.

\section{Concluding remarks}

In this paper, based on polar decomposition of multiplicative groups
of quadratic extensions of finite fields, we construct two classes
of  algebraic immunity optimal Boolean functions. We find that the
second class of Boolean functions possess almost all the necessary
properties to be used as filter functions in  stream ciphers.

In fact, in the proof of Theorem \ref{T2}, no property of the set
$\Lambda=\{1,\xi,\ldots,\xi^{2^{m-1}}\}$ has been used except the
cardinality of it. Therefore, we can construct balanced OAI Boolean
functions by setting $\Supp(F)=(\Gamma\times U) \cup
(\{1\}\times\Lambda')$ for any subset $\Lambda'$ of $U$ satisfying
$|\Lambda'|=2^{m-1}+1$. Then we have more opportunities to get
balanced OAI Boolean functions with high nonlinearity. However,
univariate representations and algebraic degrees of functions
constructed using $\Lambda'$ with no special properties would be
difficult to describe.


\end{document}